%% file: main.tex
\documentclass[11pt]{article}
\usepackage[english]{babel}
\usepackage{fullpage}
\usepackage[utf8]{inputenc}
\usepackage{amsmath,amsthm,amsfonts,mathtools}
\usepackage{xcolor}
\usepackage[hidelinks]{hyperref}
\usepackage{cleveref}
\usepackage{tcolorbox}
\usepackage{wrapfig}

\usepackage[linesnumbered,ruled,vlined]{algorithm2e}

\title{Non-obvious Manipulability in the Additively Separable Group Activity Selection Problem}
\author{Maria Fomenko$^1$, Giovanna Varricchio$^2$}
\date{%
	$^1$Gran Sasso Science Institute, L'Aquila, Italy\\%
	$^2$University of Calabria, Rende, Italy\\%
		 maria.fomenko@gssi.it, giovanna.varricchio@unical.it
	\\[2ex]%
	\today
}

\include{macros_cleaned}

\begin{document}
\maketitle

\begin{abstract}
\input{abstract}
\end{abstract}

\input{intro}
\input{preliminaries}
\input{opt_vs_nom}

\input{apx_nom}
\input{conclusions}
\bibliographystyle{plain}
\bibliography{bibliography}
\appendix
\input{appendix}

\end{document}

%% file: macros_cleaned.tex
\usepackage{amsthm}
\usepackage{enumitem}
\usepackage{cleveref}
\usepackage{marginnote}
\usepackage{todonotes}
\usepackage{multicol}
\usepackage{multirow}
\usepackage{tikz}
\usetikzlibrary{calc} 
\usetikzlibrary{arrows}
\usepackage{diagbox}
\usepackage{adjustbox}
\usepackage{graphicx}

\usepackage{subcaption}

\usepackage{microtype}
\usepackage{varwidth}

\usepackage{thm-restate}

\usepackage[numbers]{natbib}
\usepackage{times}
\usepackage{soul}
\usepackage{url}
\usepackage[hidelinks]{hyperref}
\usepackage[utf8]{inputenc}
\usepackage{graphicx}
\usepackage{amsmath}
\usepackage{amsthm}
\usepackage{booktabs}
\usepackage[switch]{lineno}

\usepackage{amsmath,amsthm,amsfonts,mathtools}

\newtheorem{theorem}{Theorem}
\newtheorem{lemma}{Lemma}
\newtheorem{corollary}{Corollary}

\theoremstyle{definition}
\newtheorem{example}{Example}
\newtheorem{definition}{Definition}

\newcommand{\agents}{\mathcal{N}}

\newcommand{\SW}{\mathsf{SW}}
\newcommand{\set}[1]{\{#1\}}

\newcommand{\modulus}[1]{\left\lvert #1 \right\rvert }

\newcommand{\opt}{\textsf{opt}}
\newcommand{\instance}{\mathcal{I}}

\newcommand{\NN}{\ensuremath{\mathbb{N}}}
\newcommand{\RR}{\ensuremath{\mathbb{R}}}

\newcommand{\classNP}{\textsf{NP}}

\newcommand{\classP}{\textsf{P}}

\newcommand{\mech}{\mathcal{M}}
\newcommand{\declared}{\mathbf{d}}

\newcommand{\agentsmini}{\agents\setminus\set{i}}

\newcounter{mechanismCounter}
\newtheorem{mecha}{Mechanism}
\newenvironment{mechanism}
{
\addtocounter{mechanismCounter}{1}
\mecha 
}
{\endmecha}

\crefname{mechanism}{Mechanism}{Mechanism}

\DeclareMathOperator*{\argmax}{arg\,max}
\DeclareMathOperator*{\argmin}{arg\,min}

\newcommand{\assignment}{\mathbf{z}}
\newcommand{\optOutcome}{\mathbf{o}}

\newcommand{\gasp}{AS-GASP}
\newcommand{\weights}{\mathcal{W}}
\newcommand{\preferences}{\mathcal{P}}
\newcommand{\nonneg}{\RR^{+}_0}
\newcommand{\arbitrary}{\RR}
\newcommand{\binary}{\set{0,1}}
\newcommand{\void}{a_{\emptyset}}

\usepackage[ruled,vlined, linesnumbered]{algorithm2e}

%% file: abstract.tex
In this work, we study the \emph{additively separable Group Activity Selection Problem} (\gasp) in an imperfect information setting, where agents have private preferences over activities and weights over other agents. 
Our goal is to design mechanisms that assign agents to activities based on their declared preferences and weights, with the objective of maximizing social welfare while ensuring truthful reporting. We, therefore, focus on the notion of \emph{non-obvious manipulability} (NOM), a form of resilience to manipulation. 
We first investigate the relationship between NOM and social welfare optimality. In this regard, our main result shows that, when preferences and weights are arbitrary or non-negative, any optimal mechanism is non-obviously manipulable. In contrast, when either preferences or weights are binary, we show that optimality and NOM may be incompatible.
We then turn to computational aspects. While it is known that computing an optimal outcome for the \gasp\ is \classNP-hard even in restricted settings, we establish a strong inapproximability result showing that no polynomial-time algorithm can guarantee a bounded approximation ratio when preferences and weights may take arbitrary values. In turn, when preferences are non-negative, we show that a bounded approximation is possible, and we present two asymptotically optimal approximation mechanisms that are also guaranteed to satisfy NOM.

%% file: intro.tex
\section{Introduction}
The \emph{Group Activity Selection Problem} (GASP)~\cite{darmann2018approval} models settings where a set of agents has to be assigned to a set of activities according to their valuations for both the activity they are assigned to and the other agents participating in the same activity. The GASP captures a wide range of realistic scenarios, including workers splitting into teams to perform tasks, employees being allocated to different locations, students assigned to classes, and many others.

An interesting subclass of the GASP is the \emph{additively separable} GASP (\gasp)~\cite{bilo2019optimality}. In this model, each agent expresses a preference for each activity and a weight for every other agent, which are real values representing her appreciation for participating in that activity and for being grouped with that agent, respectively. The utility of an agent is given by the sum of the value of the activity she is assigned to and the weights she assigns to all her co-participants. The \gasp\ generalizes a well-known class of coalition formation games, namely \emph{additively separable Hedonic Games} (ASHGs). In general, Hedonic Games (HGs) can be seen as a special case of the GASP in which the activity assigned to the agents does not influence the evaluation of the outcome, and it solely depends on the coalition members.

Most of the existing literature on the GASP and HGs focuses on stability notions, such as Nash stability and core stability, under the assumption that agents' preferences are publicly known. In contrast, this paper studies GASP in an \emph{imperfect information} setting, where agents' preferences and weights are private information. A central authority must assign agents to activities based on their declared values, with the objective of maximizing the \emph{social welfare}, i.e., the sum of agents' utilities.

In this context, the strategic aspects of how agents reveal their private information become central. The algorithm (a.k.a.\ {\em mechanism}) used to determine the assignment of agents to activities should incentivize agents to report their true preferences and weights. For example, if truthful reporting is a dominant strategy for every agent, as no agent is able to increase her utility by misreporting her private information, the mechanism is said to be {\em strategyproof}. However, recent works on strategyproof mechanisms for the \gasp\ and HGs highlight that the standard notion of strategyproofness is often incompatible with desirable outcomes \cite{flammini2021strategyproof, flammini2022strategyproof, flammini2022approximate}. For this reason, the study of \emph{non-obvious manipulability} (NOM) has been undertaken in the context of HGs, and recent findings reveal that this notion is not only more realistic for agents with limited knowledge, but also leads to more satisfactory outcomes. In line with this emerging body of work, we initiate the study of non-obvious manipulability in the context of the \gasp.

\paragraph{Our Contribution.}

\begin{table*}[t]
\centering
\begin{multicols}{2}
\resizebox{!}{0.95cm}{
  \begin{tabular}[t!]{lccc}
  \centering
        &\multicolumn{3}{c}{Weights}\\
        \hline
        Prefer. & $\arbitrary$ & $\nonneg$, $m\geq 2$ & $\binary$\\
        \hline \\[-0.3cm]
        $\arbitrary/\nonneg$ &  $\forall$ (Thm.~\ref{thm:OPTisNOMarbitraryWeights})     & $\forall$  (Thm.~\ref{thm:OPTisNOMnonNegativeWeights})     & $\not\!\exists$ (Thm.~\ref{thm:noOPTbinaryWeights})  \\[0.1cm]
        $\binary$   & $\exists$ (Thm.~\ref{thm:existsOPTbinaryPreferences})  & $\not\!\exists$ (Thm.~\ref{thm:noOPTbinaryPreferences})    & $?$  \\[0.1cm]
        \hline
    \end{tabular}
    }
       \caption{Compatibility of optimality and non-obvious manipulability. }
    \label{tab:NOMvsOPT}
\columnbreak
\resizebox{!}{0.95cm}{
\begin{tabular}[t!]{lcc}
\centering
     &\multicolumn{2}{c}{Weights}\\
        \hline
        Prefer. & $\arbitrary$ & $\nonneg$ \\
        \hline \\[-0.3cm]
        \multirow{2}*{$\nonneg$ }   & $n^2$ (Thm.~\ref{thm:matchingNOMnonNeg}+\ref{thm:quadraticAPX}) & 1, if $m\leq 2$, (\cite{bilo2019optimality} + Thm.~\ref{thm:OPTisNOMnonNegativeWeights}) 
        \\[0.1cm]
        &$\Omega(n^2)$ (Thm.~\ref{thm:quadraticInapproximability}) & $2-\frac{1}{m}$, if  $m\geq 3$,  (\cite{bilo2019optimality} + Thm.\ref{thm:2apxNOM})   \\[0.1cm]
        \hline
    \end{tabular}
    }
    \caption{Efficient approximation and non-obvious manipulability.}
    \label{tab:APXvsOPT}
\end{multicols}
\caption*{Our results. In \Cref{tab:NOMvsOPT} we show the compatibility of optimality and non-obvious manipulability. We use the notation: $\forall$ = every optimal mechanism is NOM, $\not\!\exists$ = no optimal mechanism is NOM, $\exists$ = there exists an optimal and NOM mechanism. Symbol $?$ denotes the open problems.  In \Cref{tab:APXvsOPT}, we show our approximability results in conjunction with non-obvious manipulability and poly-time complexity. }
\label{tab:results}
\end{table*}

Strategyproofness in the AS-GASP has been studied in \cite{flammini2022approximate}; however, the authors show that no bounded approximation is possible in conjunction with strategyproofness whenever preferences and weights are non-negative. 
With this work, we provide an almost complete picture of non-obviously manipulable mechanisms that guarantee good performance with respect to the utilitarian social welfare. 
We first study the compatibility of optimality and non-obvious manipulability depending on the values that preferences and weights may take.
We show that non-obvious manipulability is always compatible with optimality, except when either preferences or weights are binary. 
An overview 
of our results in this sense
of the compatibility between optimality and non-obvious manipulability 
is reported in \Cref{tab:NOMvsOPT}. 

Unfortunately, \cite{bilo2019optimality} showed that computing an optimal assignment becomes computationally hard as soon as there are at least three activities, even when preferences and weights are binary. Aiming at efficient computation, we therefore turn to the study of non-obvious manipulability in conjunction with bounded approximation. In this direction, we prove that, when preferences and weights may take arbitrary (possibly negative) real values, no polynomial-time algorithm can guarantee a bounded approximation ratio unless $\classP=\classNP$ (\Cref{thm:unboundedAPX}). Finally, by restricting attention to non-negative preferences, we present polynomial-time mechanisms that achieve both a bounded approximation guarantee and non-obvious manipulability, as summarized in \Cref{tab:APXvsOPT}. All missing proofs can be found in the appendix (\Cref{sec:supplemental}).

\paragraph{Related Work.}
In the past decade, considerable attention has been devoted to the GASP~\cite{darmann2018approval,darmann2015ordinal,darmann2017simplified,igarashi2017group,igarashi2017parameterized,ganian2023group,lee2017parameterized},  as an interesting generalization of the well-known HGs \cite{dreze1980hedonic}.  
In this stream of research, individual and group deviations have been considered, providing several hardness results concerning the existence of stable solutions along with positive results for specific cases.

In the same spirit of~\cite{aziz2011stable} for HGs, the AS-GASP 
has been introduced and studied in~\cite{bilo2019optimality}. Besides the study of Nash stable outcomes and of the price of anarchy and stability, the \classNP-hardness of maximizing the utilitarian welfare has been proven. Moreover, an efficient $(2 - 1/m)$-approximation for non-negative preferences and weights, where $m$ denotes the number of activities, was given.

Some research in HGs and GASP focused on the scenario in which agents' valuations are not known and the assignment of agents should be computed in a way that encourages truthful reporting.
Strategyproofness is a central notion for preventing strategic misreporting by agents and has been extensively studied in relation to both stability~\cite{wright2015mechanism,rodriguez2009strategy,dimitrov2004enemies} and social welfare optimality and approximation~\cite{flammini2021strategyproof,flammini2022strategyproof,flammini2022approximate}. In addition, several works analyze whether specific rules are strategyproof or manipulable~\cite{long2019strategy,darmann2019manipulability}.

In contrast to strategyproofness, in the past few years, non-obvious manipulability has been introduced~\cite{troyan2020obvious} and turned out to be a relaxation good enough to circumvent the inherent impossibility results of strategyproofness in several game-theoretical settings \cite{aziz2021obvious,troyan2024non,psomas2022fair}.

In the specific case of ASHGs, while it was shown that no bounded approximation is possible in conjunction with strategyproofness~\cite{flammini2021strategyproof};
in \cite{ferraioli2025non}, in turn, a non-obviously manipulable mechanism with the approximation ratio asymptotically matching the best polynomial approximation of the utilitarian welfare was provided. In the case of Friends and Enemies games, a special class of ASHGs, non-obvious manipulability~\cite{flammini2025non} improved the approximation linearly compared to that attained under  strategyproofness~\cite{flammini2022strategyproof}.

The interplay between strategyproofness and approximation of the utilitarian welfare for the \gasp\ has been studied in~\cite{flammini2022approximate}. However, also in this setting, strategyproofness proves to be a particularly demanding requirement, as it generally precludes the possibility of guaranteeing a bounded approximation ratio. With our work, we align with the recent stream of research on non-obvious manipulability, with the aim of achieving good approximation guarantees while still ensuring some form of resilience to manipulation.

%% file: preliminaries.tex
\section{Preliminaries}
We next present the preliminaries of our work. We use $[k]$ to denote $\set{1,\dots, k}$ for any $k\in\NN$.

\subsection{The additively separable group activity selection problem (AS-GASP)}
An \gasp\ instance $\instance=(\agents,A, \set{v_i}_{i\in\agents})$ is given by 
a set of $n\geq2$ {\em agents} $\agents$ and a set of $m\geq1$ {\em activities} $A$. Each agent $i\in\agents$ has {\em valuations} $v_i = (p_i, w_i)$ expressing the values $i$ attributes to the activities and the other agents. Specifically, $i$ defines her values towards activities with a {\em preference} function $p_i : A \rightarrow \preferences \subseteq \arbitrary$, where $\preferences$ constitutes the range of values $p_i$ may take.
Similarly, the {\em weight} function $w_i : \agents\setminus\set{i} \rightarrow \weights \subseteq \arbitrary$, with $\weights$ denoting its range of possible values, expresses the values $i$ gives to any other agent.
In this work, we consider the settings where preferences (resp.\ weights) can be {\em arbitrary} (taking values from $\arbitrary$), {\em non-negative} (taking values from $\nonneg$), or {\em binary} (taking values from $\set{0, 1}$).
For simplicity, we write $\weights\in X$, resp.\ $\preferences\in X$, to denote the admissible range of weights and of preferences, e.g., $\weights \in \set{\arbitrary, \nonneg}$ means that weights can be arbitrary or non-negative, while $\weights ={\set{0,1}}$ means that weights can only be binary.

An {\em outcome}, also called {\em assignment}, of the game will be denoted by $\assignment$, where $z_i\in A\cup \set{\void}$ is the activity of $i$ in $\assignment$ with $\void$ denoting the {\em void activity}. The void activity is used to represent the scenario in which an agent is not participating in any activity in $A$. In that case, $p_i(\void)=0$.  While the agents assigned to an~activity $a\in A$ are assumed to participate in the activity together, hence forming a {\em coalition}, an agent assigned to the void activity is considered to be alone.
Therefore, if $z_i \in A$, $\delta_i(\assignment)= \sum_{j\in \agents\,:\, z_j=z_i }w_i(j)$ denotes the overall evaluation of $i$ for the agents participating in $z_i$, otherwise, if $z_i= \void$, we set  $\delta_i(\assignment)= 0$.
The {\em utility} of agent $i$ in the assignment $\assignment$ is given by $u_i(\assignment)= p_i(z_i)+\delta_i(\assignment)$.

Given an instance $\instance$, to evaluate the performance of an assignment, we use the classical definition of {\em utilitarian social welfare}, given by $\SW^\instance(\assignment)=\sum_{i\in \agents}u_i(\assignment)$. We may simply write $\SW(\assignment)$ when the instance is clear from the context.

Given an instance $\instance$, the {\em social optimum}, denoted by $\opt(\instance)$, is the maximum utilitarian social welfare achievable by any assignment, and we denote by $\optOutcome(\instance)$ an assignment maximizing the utilitarian social welfare. When the instance is clear from the context, we refer to the social optimum and an optimal outcome simply as $\opt$ and $\optOutcome$, respectively.

\subsection{Mechanism Design Framework}
In this work, we assume that each agent $i\in \agents$ possesses a private valuation $v_i$ and, before the assignment is computed, communicates a {\em declaration} $d_i$, possibly different from 
her true valuations 
$v_i$. To distinguish  
the true values of $i$ we write $v_i=(p^t_i,w^t_i)$, where $t$ stands for {\em truthful}, and $d_i=( p_i, w_i)$ to denote any possible declaration of $i$. We denote by $\declared=(d_1, \dots, d_n)$ the collection of the declarations of all agents and by $\declared_{-i}$ the declarations of all agents except that of $i$.
Note that the utility $u_i(\assignment)$ attained by agent $i$ in a given assignment $\assignment$ is computed with respect to her true valuation $v_i$, and not with respect to the declared $d_i$. However, when computing an optimal or approximately optimal outcome, we can only rely on the agents' declarations. We write $\opt(\declared)$ and $\SW^\declared$, to denote the optimum and the social welfare function when agents declare $\declared.$

A {\em mechanism} $\mech$ is an algorithm that, for any declaration $\declared$, outputs an assignment $ \mech(\declared)$. 
We evaluate the performance of $\mech$ through the corresponding {\em approximation ratio} 
$r^\mech = \sup_\declared \frac{\opt(\declared)}{\SW^\declared(\mech(\declared))}$. A mechanism $\mech$ is {\em optimal} if $r^\mech = 1$. 

Our goal is to provide mechanisms that are both efficient in terms of the social welfare and resistant to attempts of manipulation.  
A well-studied form of resilience to manipulation is {\em strategyproofness}.
A~mechanism is said to be \emph{strategyproof} (SP) if, for any $i\in \agents$, $\declared_{-i}$, $d_i$, and true $v_i$, 
it satisfies
\[
u_i(\mech(\declared_{-i}, v_i)) \geq u_i(\mech(\declared_{-i}, d_i))\, .
\]

In turn, a mechanism is \emph{manipulable} if it is not strategyproof. 

Unfortunately, \cite{flammini2022approximate} showed that even if $\weights=\preferences=\nonneg$, no bounded approximation is possible when requiring strategyproofness. 
In order to attain good performances in the sense of the utilitarian social welfare, we consider the following relaxation.

\begin{definition}[Non-Obvious Manipulability]\label{def:NOM}
    A mechanism $\mech$ is said to be \emph{non-obviously manipulable} (NOM) if, for every $i\in \agents$, real values $v_i$, and any other declaration $d_i$, the following two conditions hold true:
    \begin{enumerate}
        \item\label{NOM:sup} $\sup\limits_{\declared_{-i}} u_i(\mech(\declared_{-i}, v_i)) \geq \sup\limits_{\declared_{-i}} u_i(\mech(\declared_{-i}, d_i))$, and 
        \item\label{NOM:inf} $\inf\limits_{\declared_{-i}} u_i(\mech(\declared_{-i}, v_i)) \geq \inf\limits_{\declared_{-i}} u_i(\mech(\declared_{-i}, d_i))$.
    \end{enumerate}
In case there exist $i$, $v_i$, and $d_i$ such that one of the aforementioned conditions is violated, $\mech$ is called \emph{obviously manipulable} and $d_i$ is an \emph{obvious manipulation}.
\end{definition}

\paragraph{Useful Notation.}

For a fixed number of agents $n$ and declaration $d_i=(p_i, w_i)$ of agent $i$, we denote by $M_{d_i}$ the value 
\begin{equation}\label{eq:M_d_i}
    n\cdot\left( 1 + \max\{\max_{j \in \agents \setminus \set{i}}\left|w_i(j)\right|,\max_{a\in A}\left|p_i(a)\right|\}\right).
\end{equation}
The value $M_{d_i}$ should be interpreted as an arbitrarily large positive real number with respect to the values declared in $d_i$.

Throughout the paper, our proofs consider an arbitrary $i\in \agents$. Consequently, without loss of generality, we may assume that the activities are sorted according to the true preferences of $i$, i.e., $p^t_i(a_1) \geq p^t_i(a_2) \geq \dots \geq p^t_i(a_m)$. 
For notational convenience, we may also assume that $a_k=\void$ for each $k\geq m+1$.

%% file: opt_vs_nom.tex
\section{Optimality and NOM}
In this section, we characterize the cases in which non-obvious manipulability is compatible with optimality, as summarized in \Cref{tab:NOMvsOPT}. To this end, we first present a series of auxiliary lemmas that explore some possible optimal outcomes when only $d_i$, for a given agent $i$, is fixed.

\subsection{Possible Optimal Outcomes for a Fixed $d_i$}

This first lemma establishes that, upon truthful reporting, there always exists a declaration of the others such that, in any optimal assignment, agent~$i$ attains her highest possible utility.

\begin{lemma}\label{lemma:OptBest}
If $\preferences, \weights \in \{\arbitrary, \nonneg, \binary\}$, then, for any truthful declaration $v_i = (p^t_i, w^t_i)$ of agent $i$, there exists a declaration of the other agents $\declared_{-i}$ such that, in any optimal outcome of the instance with $\declared = (v_i, \declared_{-i})$, agent $i$ attains her highest possible utility, namely, $\bar{u}=\max_{\assignment}u_i(\assignment)$.
\end{lemma}
\begin{proof}
Note that $\bar{u}=\max_{a\in A}p^t_i(a) + \sum_{j\neq i \,:\, w^t_i(j) > 0}w^t_i(j)$, if positive, and $\bar{u}=0$, otherwise. In fact, the utility of agent $i$, when assigned to an activity in $A$, is maximized when only the agents evaluated non-negatively participate in the same activity. Under such an assumption, the best scenario is when the selected activity is the best for $i$.  However, this assignment can still produce a negative utility for agent $i$, and in that case, it is always preferable to be assigned to the void activity. 

Let us now define the declarations $\declared_{-i}$ so that, for every agent $j \in \agentsmini$, all preferences and weights are set to $0$.

Note that for the resulting \gasp\ instance with $\declared=(v_i, \declared_{-i})$, the social welfare of any assignment $\assignment$, denoted by $\SW(\assignment)$, coincides exactly with the utility of agent $i$ in that assignment.
Therefore, if agent $i$ is assigned in $\assignment$ to some activity $a \in A$ together with a (possibly empty) subset of agents $C \subseteq \agentsmini$, we have
$  \SW(\assignment)= p_i^t(a) + \sum_{j\in C}w^t_i(j) .$

Consider now any optimal assignment $\optOutcome$.
If $\bar{u}>0$, then
in $\optOutcome$, $i$ must be assigned to any activity in $\argmax_{a\in A}p^t_i(a)$ together with all agents in $\set{j \,\vert\, w^t_i(j) > 0}$ and possibly some agents in $\set{j \,\vert\, w^t_i(j) = 0}$; no other agent is assigned to the same activity of $i$. If $\max_{a\in A}p^t_i(a) + \sum_{j: w^t_i(j) > 0}w^t_i(j)\leq 0$, 
then in any optimal outcome, agent $i$ attains zero utility, which can be achieved by assigning agent $i$ to the void activity.
\end{proof}

From \Cref{lemma:OptBest} and the definition of $\bar{u}$, we can derive:
    $\sup_{\declared_{-i}} u_i(\mech(\declared_{-i}, v_i))\! =\! \bar{u} \geq \sup_{\declared_{-i}} u_i(\mech(\declared_{-i}, d_i)).$ This shows the following.

\begin{corollary}\label{cor:OptBest}
If $\preferences, \weights \in \{\arbitrary, \nonneg, \binary\}$, any optimal mechanism $\mech$ satisfies Condition~\ref{NOM:sup} of \Cref{def:NOM} for any $i\in\agents$.
\end{corollary}

In what follows, we report additional lemmas establishing some of the possible optimal outcomes for a given $d_i$. 

\begin{restatable}{lemma}{NonemptyCoalitionArbitraryWeights}\label{lemma:NonemptyCoalitionArbitraryWeights}
If $\preferences\in \set{\arbitrary,\nonneg}$  and $\weights=\arbitrary$, given any 
$i \in \agents$, her 
declaration $d_i=(p_i, w_i)$, non-empty $C\subseteq \agentsmini$, and $a^* \in A$, there always exists declaration of other agents $\declared_{-i}$ such that in any optimum $i$ is assigned to the activity $a^*$ together with only the agents in $C$.
\end{restatable}
\begin{proof}
    Let us construct $\declared_{-i}$ as follows:
    \begin{itemize}
        \item for every $j, k \in C$, $w_j(k) = w_j(i) = p_j(a^*) = M_{d_i}$;
        \item for every $j \in C$ and $k \in \agents\setminus\left(C \cup \set{i}\right)$, $w_j(k) = w_k(j) = w_k(i) = - M_{d_i}$.
    \end{itemize}
    Any other preference and weight is set to $0$.  
    
    By the definition of $M_{d_i}$, see Eq.~(\ref{eq:M_d_i}), in any optimal assignment, the agents in $C$ as well as $i$ must be assigned to $a^*$, while no other agent is assigned to $a^*$. 
\end{proof}

The proofs of the remaining lemmas are almost straightforward and rely on a similar constructions of $\declared_{-i}$. Detailed proofs are provided in the appendix.

\begin{restatable}{lemma}{WorstActivity}\label{lemma:WorstActivity}
    If $\preferences \in \set{\nonneg, \arbitrary}$ and $\weights = \arbitrary$, then, for any declaration $d_i = (p_i, w_i)$ of agent $i$, and any subset of activities $A' \subseteq A$ with $\modulus{A'} = n-1$, if $m \geq n$, and $A' = A$ otherwise, there exists a declaration of the other agents $\declared_{-i}$ such that, in any optimum, agent $i$ is assigned to some activity $a \in \left(A\setminus A'\right) \cup \set{\void}$ alone.
\end{restatable}

\begin{restatable}{lemma}{OptSingle}\label{lemma:OptSingle} 
If $\weights\in\set{\nonneg,\binary}$,  given any declaration $d_i=(p_i, w_i)$, in any optimum an agent $i$ is the only agent assigned to an activity $a$ only if $a\in \argmax_{a\in A\cup\set{a_{\emptyset}}}p_i(a)$.
\end{restatable}

\begin{restatable}{lemma}{NonemptyCoalition}\label{lemma:NonemptyCoalition}
    If $\preferences \in \set{\arbitrary, \nonneg}$, $\weights=\nonneg$ and $m\geq 2$, given any 
    agent $i \in \agents$, her 
    declaration $d_i=(p_i, w_i)$, non-empty $C\subseteq \agentsmini$, and $a^* \in A$, there exists declaration of other agents $\declared_{-i}$ such that in any optimum $i$ is assigned to $a^*$ with only the agents in $C$.
\end{restatable}

\begin{restatable}{lemma}{AloneToActivity}\label{lemma:AloneToActivity}
    If $\preferences=\arbitrary$ and $\weights\in\set{\arbitrary,\nonneg}$, then given any 
    agent $i \in \agents$ and any 
    declaration $d_i = (p_i, w_i)$ of agents $i$, there exists a declaration of the other agents $\declared_{-i}$ such that, in any optimum, agent $i$ is assigned to some activity $a^*\in A\cup\set{\void}$ alone.
\end{restatable}

\subsection{Optimality and NOM for Arbitrary Weights}
We are ready to explore the compatibility of optimality and non-obvious manipulability. We start with the case $\weights=\arbitrary$.

\begin{theorem}\label{thm:OPTisNOMarbitraryWeights}
If $\preferences\in\set{\arbitrary, \nonneg}$ and $\weights=\arbitrary$, any optimal mechanism $\mech$ satisfies NOM.
\end{theorem}
\begin{proof}
Let $i$ be an arbitrary agent with truthful valuations $v_i=(p_i^t, w_i^t)$. Recall that we can assume
$p_i^t(a_1) \geq p_i^t(a_2) \geq \dots \geq p_i^t(a_m)$.  Moreover, for all $k \geq m+1$, $a_k = \void$.

Condition~\ref{NOM:sup} of \Cref{def:NOM} follows from \Cref{cor:OptBest}. We next prove Condition~\ref{NOM:inf} of \Cref{def:NOM}. 

Let $\bar{p}=\max\set{p_i^t(a_n), 0}$ and, denoted by $w^t_i(C)= \sum_{j\in C} w^t_i(j)$, where $C\subseteq \agentsmini$, let $C^*$ be any coalition s.t.
$$ C^* \in \argmin_{{\emptyset \neq C\subseteq\agentsmini}} w^t_i(C).$$

We first focus on the truthful declaration $v_i$ of $i$ and show 
\begin{align}\label{eq:LBinfArbitraryWeights}
    \inf\nolimits_{\declared_{-i}} u_i(\mech(\declared_{-i}, v_i)) \geq \min\set{p_i^t(a_m) + w_i^t(C^*), \bar{p}} \,.
\end{align}
For the sake of contradiction, suppose that there is an optimal outcome where $i$ gets the utility strictly lower than $u^{min}=\min\set{p_i^t(a_m) + w_i^t(C^*), \bar{p}} $. 
If $i$ is assigned to some activity $a\in A$ with only the agents from $C \subseteq \agentsmini$, where $C\neq \emptyset$, then, by definition of $C^*$ and the ordering of activities, we have
$w_i^t(C) \geq w_i^t(C^*)$ and $p_i^t(a) \geq p_i^t(a_m)$. Hence, such an assignment yields utility at least $w_i^t(C^*) + p_i^t(a_m)\geq u^{min}$ for agent $i$.
Therefore, to obtain a utility lower than $u^{min}$, $i$ must be assigned alone to some activity in $A \cup \set{\void}$. Assume that $i$ is assigned to $a_k$, hence, $p^t_i(a_k)< u^{min}\leq \bar{p}$.
Since the outcome is optimal, it must be that $p_i^t(a_k) \geq 0$; otherwise, assigning agent $i$ to the void activity would strictly increase the social welfare. This implies that $\max\{p_i^t(a_n), 0\}= \bar{p}> p^t_i(a_k) \geq 0$, and hence $p_i^t(a_n) > 0$. Therefore, $n \leq m$, as otherwise $a_n = \void$ and $ p_i^t(a_n) =0$. Moreover, $k > n$, as $p_i^t(a_n) > p_i^t(a_k)$. 
On the other hand, since there are only $n$ agents, there must exist some $h \leq n$ such that no agent is assigned to activity $a_h$. We can then reassign agent $i$ from $a_k$ to $a_h$, which would strictly increase the social welfare, as
$p_i^t(a_h) \geq p_i^t(a_n) > p_i^t(a_k)$. This contradicts the optimality of the original outcome and proves \Cref{eq:LBinfArbitraryWeights}.

Consider now any possible declaration $d_i$ of $i$, we next show 
\begin{align}\label{eq:UBinfArbitraryWeights}
    \inf\nolimits_{\declared_{-i}} u_i(\mech(\declared_{-i}, d_i)) \leq \min\set{p_i^t(a_m) + w_i^t(C^*), \bar{p}}\, .
\end{align}

By \Cref{lemma:NonemptyCoalitionArbitraryWeights}, for any $d_i$, there exists a $\declared_{-i}$ such that $i$ is assigned to one of her least preferred activities, without loss of generality $a_m$, together with exactly all the agents in $C^*$. Hence, $\inf\nolimits_{\declared_{-i}} u_i(\mech(\declared_{-i}, d_i)) \leq p_i^t(a_m) + w_i^t(C^*)$.

If we set $A' = \{a_1, \dots, a_{n-1}\}$ and apply \Cref{lemma:WorstActivity}, there exists $\declared_{-i}$ such that $i$ is assigned to an activity $a_h$, with $h \geq n$, alone, implying that agent $i$ gets a utility $p_i^t(a_k)\leq \max\set{p_i^t(a_n), 0}$. This shows $\inf\nolimits_{\declared_{-i}} u_i(\mech(\declared_{-i}, d_i)) \leq \bar{p}$.

These two facts together prove \Cref{eq:UBinfArbitraryWeights}.

In conclusion, putting together \Cref{eq:LBinfArbitraryWeights,eq:UBinfArbitraryWeights},  Condition~\ref{NOM:inf} of \Cref{def:NOM} holds true.
\end{proof}

It remains to consider the case where $\preferences=\binary$. Let us first show that not every optimal mechanism satisfies NOM.

\begin{example}
Assume $\preferences=\binary$ and $\weights=\arbitrary$. Assume we have an instance with two activities, $a_1$ and $a_2$, and two agents. Consider an optimal mechanism $\mech$ that uses the following tie-breaking rule: When $1$ (resp.\ $2$) values $1$ only $a_1$  (resp.\ $a_2$) and the sum of the respective weights is  $\geq 1$, then, if  $1$ values $0$ agent $2$, assign both the agents to $a_2$, otherwise, assign the agents to $a_1$. When $1$ values $1$ only $a_1$, and the sum of the respective weights is $<1$, only agent $1$ is assigned to $a_1$, the assignment of agent $2$ depends on her preferences. Note that the tie-breaking does not affect the optimality of the outcome.

Let agent $1$ have the following true values: $w^t_1(2) = 0$, $p^t_1(a_1) = 1$, $p^t_1(a_2) = 0$. The worst possible utility of the agent $1$ equals $0$. This may happen when agent $2$ declares $w_2(1) = 1$, $p_2(a_2) = 1$, $p_2(a_1) = 0$, $\mech$ will assign both agents $1$ and $2$ to $a_2$ because of the tie-breaking rule. 

Now, let us consider the declaration $d_1=(p_1, w_1)$ of agent $1$ identical to the true one except for $w_1(2)$, which is declared to be $0.1$. Then, regardless of $d_2$, $1$ will always be assigned to $a_1$ by the mechanism. In fact, if $w_1(2) + w_2(1)\geq 1$, both are assigned to $a_1$, otherwise, only agent $1$ is assigned to $a_1$.

Therefore, an optimal mechanism that uses the aforementioned tie-breaking rule does not satisfy NOM.
\end{example}

Despite the fact that there might exist optimal mechanisms that are obviously manipulable. With an appropriate definition of a tie-breaking rule, we can show the following.

\begin{restatable}{theorem}{existsOPTbinaryPreferences}\label{thm:existsOPTbinaryPreferences}
    If $\preferences=\binary$, $\weights=\arbitrary$, and $m \geq 2$, there exists an optimal mechanism that satisfies NOM.
\end{restatable}

The proof of the theorem relies on the selection of the tie-breaking rule and follows a similar approach of the proof \Cref{thm:OPTisNOMarbitraryWeights}; we refer to the appendix for the details.

\subsection{Optimality and NOM for Non-negative Weights}
In \cite{flammini2022approximate} it has been shown that when $\preferences=\nonneg$, $\weights=\nonneg$ and $m=1$, any optimal mechanism is SP. We next focus only on the case where $m\geq 2$.

\begin{restatable}{theorem}{OPTisNOMnonNegativeWeights}\label{thm:OPTisNOMnonNegativeWeights}
    If $\preferences\in\set{\arbitrary, \nonneg}$, $\weights=\nonneg$, and $m\geq 2$, any optimal mechanism $\mech$ satisfies NOM.
\end{restatable}
\begin{proof}[Proof of the case $\preferences=\arbitrary$.]
Here we assume $\preferences=\arbitrary$ and defer the case where $\preferences=\nonneg$ to the supplemental material.

    Condition~\ref{NOM:sup} of \Cref{def:NOM} follows from \Cref{cor:OptBest}. We next show that Condition~\ref{NOM:inf} of \Cref{def:NOM} is also satisfied.

    Let $ j^* \in \argmin_{j \in \agents} w^t_i(j)$. Recall can assume $a_1$ is the most preferred while $a_m$ is the least preferred activity according to $i$'s truthful declaration. We define $\tilde{p}= \max\set{p^t_i(a_1), 0}$.

    We first consider the truthful report $v_i=(p^t_i, w^t_i)$ and show    \begin{align}\label{eq:LBinfNNweights}
       \inf\nolimits_{\declared_{-i}} u_i(\mech(\declared_{-i}, v_i)) \geq \min\set{p_i^t(a_m) + w_i^t(j^*), \tilde{p}}\, .
    \end{align}
    Since $\weights = \nonneg$, and $a_m$ is the least preferred activity of $i$, and, by the definition of $j^*$, for any $\emptyset \neq C \subseteq \agentsmini$ and $a \in A$, it holds that $p_i^t(a) + w_i^t(C)\geq p_i^t(a_m) + w_i^t(j^*)$, as soon as $i$ is assigned to some activity with at least one other agent, she has a utility of at least $p_i^t(a_m) + w_i^t(j^*)$. 
   To obtain a lower utility, agent $i$ must be assigned alone to some, possibly void, activity. From \Cref{lemma:OptSingle}, such an activity can only be an activity in $\argmax_{a\in A \cup \set{\void}} p_i^t(a)$, which means an activity having value $p_i^t(a_1)$, if its value is non-negative, otherwise, $i$ will be assigned to the void activity. Therefore, if $i$ is assigned alone to some activity, she gets a utility of at least $\tilde{p}$.
   
In conclusion, \Cref{eq:LBinfNNweights} holds true.

    Let us now consider an arbitrary declaration $d_i$ of agent $i$. Let $C = \set{j^*}$ and $a^*= a_m$. By \Cref{lemma:NonemptyCoalition}, there exists a declaration of the other agents such that, in any optimal assignment, agent $i$ is assigned to activity $a_m$ together with exactly agent $j^*$.
    By \Cref{lemma:AloneToActivity}, for any $d_i$ there exists  $\declared_{-i}$ such that $i$ is assigned alone to some, possibly void, activity. 
    This means she can always get a utility lower than or equal to $\tilde{p}$.
    Hence, when declaring an arbitrary $d_i$ we have
    \begin{align}\label{eq:UBinfNNweights}
    \inf\nolimits_{\declared_{-i}} u_i(\mech(\declared_{-i}, d_i)) \leq \min\set{p_i^t(a_m) + w_i^t(j^*), \tilde{p}} \, .
    \end{align}
    
\Cref{eq:LBinfNNweights,eq:UBinfNNweights} show Condition~\ref{NOM:inf} holds true.
\end{proof}

In contrast, when $\preferences=\binary$, we have the following:
\begin{restatable}{theorem}{noOPTbinaryPreferences}\label{thm:noOPTbinaryPreferences}
    If $\preferences=\binary$, $\weights=\nonneg$, and $m \geq 2$, no optimal mechanism satisfies NOM.
\end{restatable}

\begin{proof}
  Let us consider a game instance with two activities, $a_1$ and $a_2$, and three agents. Let  the truthful values of agent $3$ be as follows: $w^t_3(1) = w^t_3(2) = \varepsilon$, with  $0 < \varepsilon < \frac{1}{2}$, and $p^t_3(a_1) = p^t_3(a_2) = 1$.
    Since any activity would give this agent the utility $1$ and she has positive weights for both other agents, in any optimum, she gets at least $1 + \varepsilon$. 
    Moreover, such a utility is attained when $3$ is reporting truthfully: if other agents declare $w_1(3) = w_2(3) = 0$, $w_1(2) = w_2(1) = 0$, $p_1(a_1) = 1$, $p_1(a_2) = 0$, $p_2(a_1) = 0$, $p_2(a_2) = 1$, then there are two optima. In the first one, $1$ and $3$ are assigned to $a_1$, while $2$ is assigned to $a_2$, and in the second one, $2$ and $3$ are assigned to $a_2$, while $1$ is assigned to $a_1$. In both of the cases, the agent $3$ gets the utility $1 + \varepsilon$.

    Now, let us consider a declaration $d_3$ of $3$: $w_3(1) = w_3(2) = 4$, and $p_3(a_1) = p_3(a_2) = 1$. Here, regardless of the declarations of $1$ and $2$, in any optimum $1$, $2$, and $3$ will be assigned to the same activity (since $1$ and $2$ cannot have negative weights and their preferences over activities cannot exceed $1$). Therefore, in any optimum, agent $3$ will have the utility $1 + 2\varepsilon$, which means that $d_3$ improves the worst case of $3$ and thus is an obvious manipulation.
\end{proof}

\subsection{Optimality and NOM for Binary Weights}
In this section, we consider binary weights. 
\begin{restatable}{theorem}{noOPTbinaryWeights}\label{thm:noOPTbinaryWeights}
    If $\preferences\in\set{\arbitrary,\nonneg}$, $\weights=\binary$, and $m \geq 2$, no optimal mechanism satifies NOM.
\end{restatable}

The case of binary preferences and weights remains open. In  \Cref{sec:furtherResults} we show that not all optimal mechanism satisfies NOM; however, it is unknown whether, with a suitable tie-breaking rule, optimality and NOM can be made compatible.

%% file: apx_nom.tex
\section{Approximate Mechanisms}
In the previous section, we showed that as long as preferences and weights take values in $\arbitrary$ or $\nonneg$, any optimal mechanism is NOM\footnote{It remained open the case where $m=1$, $\preferences=\arbitrary$ and $\weights=\nonneg$.}. Unfortunately, \cite{bilo2019optimality} showed that computing an optimal outcome for the \gasp\ is \classNP-hard for $m \geq 3$, even when $\preferences = \weights = \binary$. 

In this section, we turn to the approximability of the problem, possibly in conjunction with non-obvious manipulability.

\subsection{Arbitrary Preferences and Weights}
We begin by establishing a strong inapproximability result for the general problem. 
In \cite{bilo2019optimality}, the authors focused only on the case where all values are non-negative, providing a constant approximation. We next show that, in contrast, no bounded approximation is possible when preferences and weights may take negative values.
Our proof is based on a reduction from 
the \textsc{MaxClique} problem which is \classNP-hard~\cite{karp2009reducibility} and can be formulated as follows:
{
\begin{center}
\fbox{
\begin{varwidth}{0.95\linewidth}
\noindent {\bf \textsc{MaxClique}\ problem}

\smallskip
\noindent\textit{Input:} An undirected and unweighted graph $G=(V,E)$.

\smallskip
\noindent\textit{Problem:} To find a largest subset $V'\subseteq V$ such that the induced graph of $V'$ is a clique.
\end{varwidth}
}
\end{center}
}

\begin{theorem}\label{thm:unboundedAPX}
If $\preferences = \arbitrary$ and $\weights = \arbitrary$, then no polynomial-time algorithm can guarantee a bounded approximation of $\opt$, even when $m = 1$, unless $\classP=\classNP$.
\end{theorem}
\begin{proof}
Let us consider an instance of the \textsc{maxClique} problem, which consists of an undirected graph $G = (V, E)$. We can assume that $E\neq \emptyset$ and $V=\set{1,\dots, n}$. 

For any fixed $k\leq \modulus{V}$, we construct an \gasp\ instance $\instance^k$ as follows. We have a set of agents $\agents = V$, and we introduce a single activity $a$. We also set:
\begin{itemize}
    \item for every $i, j \in \agents$ with $i\neq j$: $w_i(j) = w_j(i) = 1$, if $\set{i, j} \in E$, and $w_i(j) = w_j(i) = -n^3$, otherwise;
    \item for each $i \in \agents$: $p_i(a) = - (k -2)$, with $k\in \NN$.
\end{itemize}
Consider now any outcome $\assignment$. Since there exists a unique activity $a$, there also exists a unique group of agents $C\subseteq\agents$ such that the agents in $C$ are assigned to $a$ in $\assignment$, while the remaining agents $\agents\setminus C$ are assigned to $\void$.
Observe that if there exist two distinct $i,j\in C$ with $(i,j)\not\in E$, then,  $\SW(\assignment)< 0$. 
Hence, an outcome $\assignment$ can have positive social welfare only if the set of agents in $C$ form a clique in $G$. As a consequence, with $c$ denoting the size of $C$, we have $\SW(\assignment) = -(k-2)\cdot c+ c(c-1) $; therefore, $\SW(\assignment)>0 $ if and only if $c> k -1$, which means, if and only if there exists a clique of size at least $k$ in $G$.

Assume we have a poly-time algorithm with bounded approximation. On the instance $\instance^k$, it assigns the agents of some non-empty $C\subseteq \agents$ to $a$, if and only if a clique of size at least $k$ exists in $G$, otherwise, all the agents are assigned to $\void$. Then, we can use such an algorithm to determine in poly-time what the size of the largest clique in $G$ is by testing all $k\in [n]$. This is not possible unless $\classP = \classNP$.
\end{proof}

\subsection{Non-negative Preferences and Arbitrary Weights}
Given the negative result above, we now show that as soon as $\preferences=\nonneg$ and $\weights=\arbitrary$, a bounded approximation is possible. In particular, we present an asymptotically optimal approximation mechanism that is also guaranteed to be NOM.

\begin{mechanism}\label{mech:matching}
Informally, at each step, the mechanism selects either a single agent or a pair of agents that have not yet been assigned to any activity and assigns them to an activity that has not yet been used. The selection is made greedily, in the sense that among all available options, the one that yields the greatest marginal improvement to the utilitarian welfare is chosen. This process continues as long as there remain unassigned pairs of agents and unused activities. If some agents remain unassigned, while all the activities have been used, they will be assigned to the void activity. 

A formal description of the mechanism 
is given in Algorithm~\ref{algo:matching}. Note that, when we use the notation $x\gets X$, for some set $X$ and an element $x\in X$, an element $x$ is chosen arbitrarily from $X$.
\end{mechanism}

\begin{algorithm}[t]
	\SetNoFillComment
	\DontPrintSemicolon
	\KwIn{An \gasp\ instance $\instance=(\agents, A, \declared)$}
	\KwOut{An assignment $\assignment$}
    $\agents'\gets \agents$, $A'\gets A$\\[0.3em]
    Initialize $\assignment$ so that $z_i=\void$, for all $i\in\agents$\\[0.3em]
    \While{$\modulus{\agents'}\geq 2\wedge A'\neq\emptyset$}{
    $x\gets \max\limits_{i \in \agents', a \in A'} p_i(a)$\\[0.3em]
    $y\gets \max\limits_{i, j\in \agents', a \in A'}w_i(j) + w_j(i) + p_i(a) + p_j(a)$\\[0.3em]
    \uIf{$x\geq y$}{
    \!\!$i,a\gets \argmax\limits_{i \in \agents', a \in A'} p_i(a)$ \label{line:argmax1}
     and 
    $z_i\gets a$\\[0.3em]
    \!\!$\agents'\gets\agents'\setminus\set{i}$, $A'\gets A'\setminus\set{a}$
    }\Else{
    $ \!\!i, j, a \gets \!\!\!\!\argmax\limits_{i, j\in \agents', a \in A'}w_i(j) + w_j(i) + p_i(a) + p_j(a)$\label{line:argmax2}\\[0.3em]
    \!\!$z_i\gets a$, $z_j\gets a$\\[0.3em]
    \!\!$\agents'\gets\agents'\setminus\set{i,j}$, $A'\gets A'\setminus\set{a}$
    }
    }
    \If{$\modulus{\agents'}=1\wedge A'\neq\emptyset$}{
    $i,a\gets \argmax\limits_{i \in \agents', a \in A'} p_i(a)$  
    and 
    $z_i\gets a$\\[0.3em]
    }
    \KwRet{$\assignment$}\\
    \caption{\Cref{mech:matching}.\label{algo:matching}}
\end{algorithm}

\begin{theorem}\label{thm:matchingNOMnonNeg}
If $\preferences=\nonneg$ and $\weights=\arbitrary$, \ref{mech:matching} is NOM.
\end{theorem}
\begin{proof}
We prove Conditions~\ref{NOM:sup} and~\ref{NOM:inf} of \Cref{def:NOM} separately for a fixed agent $i$ with truthful values $v_i=(p^t_i, w^t_i)$. 

Let us observe that in any outcome of \Cref{mech:matching}, there are three possible scenarios for a fixed agent~$i$:
(i) being assigned to an activity $a \in A$ together with another agent $j \in \agents\setminus\set{i}$;
(ii) being assigned to an activity $a \in A$ alone; or
(iii) being assigned to the void activity $a_{\emptyset}$.

{\em Condition~\ref{NOM:sup}.} 
Note that, independently of $d_i$, the highest utility she can achieve in an outcome of \ref{mech:matching} is $\max_{a\in A}p^t_i(a) + \max_{j\in\agents\setminus \set{i}}w^t_i(j)$, if there exist $j\neq i$ such that $w^t_i(j) \geq 0$, or $\max_{a\in A}p^t_i(a)$, if $w^t_i(j)<0$ for each $j\neq i$.

When $i$ declares truthfully, there exists a declaration of the others where \ref{mech:matching} guarantees such a utility to agent $i$. This happens, for example, when any other weight and preference of any other agent $j$ is set to $0$. 
Therefore, the best case cannot be improved by any misreport of agent $i$.
    
{\em Condition~\ref{NOM:inf}.}  Let us consider 
$j^* \in \argmin_{j \in \agents \setminus \set{i}} w^t_i(j).$
Recall we can assume $A= \set{a_1, \dots, a_m}$ and $p^t_i(a_1)\geq \dots \geq p^t_i(a_m)$; moreover, $a_k=\void$ for any $k>m$. We also set $\bar{p}=p^t_i(a_n)$

In any outcome of \Cref{mech:matching}, when truthfully reporting, the following lower bound of $i$'s utility  holds:
\begin{align}\label{eq:LBinfArbitraryWeightsMatching}
    \inf\nolimits_{\declared_{-i}} u_i({\ref{mech:matching}}(\declared_{-i}, v_i)) \geq \min\set{p_i^t(a_m) + w_i^t(j^*), \bar{p}}     
\end{align}

In fact, if $i$ is assigned by \ref{mech:matching} to some activity $a$ together to some agent $j$, then, by definition of $j^*$ and $a_m$, $p^t_i(a) + w^t_i(j) \geq p^t_i(a_m) + w^t_i(j^*)$. If instead $i$ is assigned alone to some activity, say $a_k$, for the sake of a contradiction assume, $\bar{p}> p^t_i(a_k) $, then, $\bar{p}=p_i^t(a_n)> p^t_i(a_k) \geq 0$, and hence $p_i^t(a_n) > 0$. Therefore, $n \leq m$, as otherwise $a_n = \void$ and $ p_i^t(a_n) =0$. Moreover, $k > n$, as $p_i^t(a_n) > p_i^t(a_k)$. 
On the other hand, since there are only $n$ agents, there must exist some $h \leq n$ such that no agent is assigned to activity $a_h$ in the outcome of \Cref{mech:matching}. This implies,
$p_i^t(a_h) \geq p_i^t(a_n) > p_i^t(a_k)$, which is not possible because of the greedy choice of $a_k$ during the execution of \ref{mech:matching}. Therefore, $p_i^t(a_k) \geq p^t_i(a_n)= \bar{p}$. This shows \Cref{eq:LBinfArbitraryWeightsMatching}. 

Consider now an arbitrary  $d_i=(p_i,w_i)$. We next prove
\begin{align}\label{eq:UBinfArbitraryWeightsMatching}
    \inf\nolimits_{\declared_{-i}} u_i({\ref{mech:matching}}(\declared_{-i}, d_i)) \leq \min\set{p_i^t(a_m) + w_i^t(j^*), \bar{p}}\, .    
\end{align}

Given any $d_i$,
\Cref{mech:matching} may possibly assign $i$ to $a_n$, and hence provide a utility of $p^t_i(a_n)= \bar{p}$ to agent $i$. This can happen when the declarations of others are as follows:
\begin{itemize}
    \item for $j\in [i-1]$, $p_j(a_j) = M_{d_i}$;
    \item for $j \in [m+1]\setminus [i]$, $p_{j+1}(a_j) = M_{d_i}$;
    \item for $j \in \agents\setminus\set{i}$, $w_j(i) = -M_{d_i}$.
\end{itemize}
Any other preference and weight is set to $0$.  

Moreover, given any $d_i$,
\Cref{mech:matching} may possibly assign $i$ to $a_m$ together with $j^*$. In fact, if $\declared_{-i}$ is such that $w_{j^*}(i) = M_{d_i}$, $p_{j^*}(a_m) = M_{d_i}$, and all other preferences and weights are set to $0$, such an outcome is indeed possible. This guarantees $i$ a utility of $p_i^t(a_m) + w_i^t(j^*)$.

These two cases guarantee that \Cref{eq:UBinfArbitraryWeightsMatching} holds true.

Finally, \Cref{eq:LBinfArbitraryWeightsMatching,eq:UBinfArbitraryWeightsMatching} prove Condition~\ref{NOM:inf} of \Cref{def:NOM} and the theorem follows.
\end{proof}

In the appendix (\Cref{appendix:binaryMechanism}), we discuss the case $\preferences=\binary$. Interestingly, \Cref{mech:matching} may fail to satisfy NOM unless the tie-breaking rule is chosen appropriately.

\begin{theorem}\label{thm:quadraticAPX}
If $\preferences\in\set{\nonneg,\binary}$ and $\weights=\arbitrary$, \Cref{mech:matching} provides an $n^2$approximation of the social optimum.
\end{theorem}
\begin{proof}
Let us consider the first step of \Cref{mech:matching}. Whichever of the two options is chosen, it provides a surplus in the social welfare of at least $\max\set{p_i(a), w_j(k)}$, for any $i, j, k \in \agents$ and $a \in A$. Since in any iteration \ref{mech:matching} the social welfare never decreases, for any $\declared,$ it holds  
\begin{align*}
    \SW\left(\ref{mech:matching}(\declared)\right) \geq \max\set{\max_{i \in \agents, a \in A} p_i(a), \max_{j, k \in \agents} w_j(k)} \,.
\end{align*}

On the other hand, we have $\opt(\declared) \leq W + P$, where\\ $W = \sum_{j, k \in \agents} w_j(k)$ and $P = \sum_{i \in \agents, a \in A} p_i(a)$, and
\begin{align*}
    W + P
    &\leq n(n-1)\cdot\max_{j, k \in \agents} w_j(k) + n\cdot\max_{i \in \agents, a \in A} p_i(a)  \, .
\end{align*}
In conclusion,  $\opt(\declared)\leq n^2\cdot\SW\left(\ref{mech:matching}(\declared)\right)$ for any $\declared$.
\end{proof}

We next prove that the result is tight asymptotically.
\begin{restatable}{theorem}{quadraticInapproximability}\label{thm:quadraticInapproximability}
    If $\preferences\in\set{\nonneg, \binary}$ and $\weights=\arbitrary$,  no algorithm can provide an approximation $O(n^{2-\epsilon})$ in polynomial time even if $m=1$, unless $\classP=\classNP$. 
\end{restatable}
To show the theorem, we use the fact that for \textsc{MaxClique}, there exists no poly-time algorithm providing a $O(\modulus{V}^{1-\epsilon})$-approximation, for any $\epsilon\ll1$~\cite{haastad1999clique}.

\subsection{Non-negative Preferences and Weights}
When $\preferences=\weights=\nonneg$, Bilò et al. showed that for $m\leq 2$ an optimal solution can be computed in poly-time, and we showed that any opt is NOM. Bilò et al.\ also provided a $2-\frac{1}{m}$ approximation for $m\geq 3$. We next reformulate their algorithm and show that it is indeed a NOM mechanism for this setting. 

\begin{mechanism}\label{mech:2appx}
Given $\declared$, the mechanism at first creates the following two assignments:
\begin{enumerate}
    \item[$\assignment^1$] where all agents are assigned to the same activity $a^*\in\argmax_{a\in A} \sum_{i\in\agents} p_{i}(a)$, that is, $z_i^1=a^*$, $\forall\,i\in\agents$;
    \item[$\assignment^2$] where each agent is assigned to her most preferred activity, that is, $z_i^2\in\argmax_{a\in A} p_i(a)$, $\forall\,i\in\agents$.
\end{enumerate}
In both cases, ties are broken according to a prefixed ordering of the activities.
Then, \ref{mech:2appx} outputs $\assignment^1$, if $\SW(\assignment^1)\geq \SW(\assignment^2)$, and $\assignment^2$, otherwise.
\end{mechanism}

\begin{restatable}{theorem}{twoapxNOM}\label{thm:2apxNOM}
    If $\preferences,\weights= \nonneg$, for $m\geq 3$, \ref{mech:2appx} is NOM.
\end{restatable}
\begin{proof}
We prove Conditions~\ref{NOM:sup} and~\ref{NOM:inf} of \Cref{def:NOM} separately for a fixed agent $i$ with truthful values $v_i=(p^t_i, w^t_i)$. 

{\em Condition~\ref{NOM:sup} (sketch).}
    If $i$ truthfully reports $v_i=(p_i^t, w_i^t)$, it is not hard to see that there always exists a declaration of the others guaranteeing her a utility equal to $\max_{a\in A} p_i^t(a) + \sum_{j\in\agents\setminus\set{i}} w_i^t(j)$. The appropriate declaration is provided in the appendix for completeness. 
    Since no assignment can guarantee $i$ higher utility, no manipulation can increase her utility in the best-case scenario and Condition~\ref{NOM:sup} follows. 
    
{\em Condition~\ref{NOM:inf}.}
    Among the outcomes of type $\assignment^1$, the worst possible for $i$ is being assigned to an $a^-\in \argmin_{a\in A}p_i^t(a)$, together with all the other agents; let us denote this specific outcome $\bar{\assignment}^1$.
    On the other hand, among the outcomes of type $\assignment^2$, the worst possible for agent $i$ is one where she is assigned to $a^+\in \argmax_{a\in A} p_i^t(a)$ alone and all the other agents are assigned to $a\neq a^+$; let us denote this specific outcome $\bar{\assignment}^2$.

    Assume $u_i(\bar{\assignment}^1) \leq u_i(\bar{\assignment}^2)$,  hence, $\bar{\assignment}^1$ is the worst outcome for $i$ that \Cref{mech:2appx} may possibly provide. Given $d_i=(p_i, w_i)$,
    let $a^-\in \argmin_{a\in A}p_i^t(a)$, we define $\declared_{-i}$ as follows:
    \begin{itemize}
        \item for all $j\in\agents\setminus\set{i}$, $p_j(a^-)=M_{d_i}$;
        \item for all $j\in\agents\setminus\set{i}$ and $k\in\agents$, $w_j(k)=M_{d_i}$.
        \item any other value is set to $0$.
    \end{itemize}

    In such an instance, the algorithm will assign all the agents to $a^-$, regardless of the declaration $d_i$ of $i$. Therefore, we can conclude $\inf\nolimits_{\declared_{-i}} u_i(\mech(\declared_{-i}, d_i))=\inf\nolimits_{\declared_{-i}} u_i(\mech(\declared_{-i}, v_i))$.

    Assume now $u_i(\bar{\assignment}^1) > u_i(\bar{\assignment}^2)$, hence, $\bar{\assignment}^2$ is the worst possible outcome that \Cref{mech:2appx} may possibly provide for agent $i$.
    Let $a^*\in \argmax_{a\in A} p_i(a)$ the activity coming first in the prefixed ordering used in \Cref{mech:2appx}, that is, when $i$ declares $d_i=(p_i, w_i)$ in the assignment of type $\assignment^2$ she will be assigned to the activity $a^*$.
    Now, let $a', a''\in A\setminus\set{a^*}$, and let $j^*\in\agentsmini$, we set $\declared_{-i}$ as follows:
    \begin{itemize}
        \item for all $j\in\agents\setminus\set{i,j^*}$, $p_j(a')=M_{d_i}$;
        \item $p_{j^*}(a'')=M_{d_i}$;
        \item any other value is set to $0$.  
    \end{itemize}
    In this instance, the assignment of type $\assignment^2$ guarantees a higher welfare than that of $\assignment^1$. 
    Moreover, in the assignment of type $\assignment^2$, $i$ is assigned to $a^*$. Hence, when declaring $d_i$ we have $\inf\nolimits_{\declared_{-i}} u_i(\mech(\declared_{-i}, d_i))\leq p^t_i(a^*) \leq  p^t_i(a^-) = u_i(\bar{\assignment}^2)$. Being $\bar{\assignment}^2$ the worst possible for $i$,  $\inf\nolimits_{\declared_{-i}} u_i(\mech(\declared_{-i}, v_i))\geq u_i(\bar{\assignment}^2)$. This proves that Condition~\ref{NOM:inf} holds true also in this case.
\end{proof}

%% file: conclusions.tex
\section{Future Work}
Although we provided an almost complete picture of the problem, several interesting open questions remain. 

First, the compatibility between optimality and non-obvious manipulability in the case of binary preferences and binary weights is still unresolved. In the supplemental material, we present an example showing that not every optimal mechanism satisfies NOM, and we conjecture that there exists at least one optimal mechanism that is non-obviously manipulable.

Moreover, \cite{ferraioli2025non} provided, for additively separable and fractional HGs, a characterization of sufficient conditions under which approximate mechanisms satisfy NOM. Extending such a characterization to the \gasp\ appears a promising direction for future work.

Finally, from an approximation point of view, the complexity of approximation for $\preferences=\nonneg$ and $\weights=\arbitrary$ remains open. 

%% file: appendix.tex
\section{Appendix}\label{sec:supplemental}

\subsection{Missing Proofs}
\WorstActivity*
\begin{proof}
    Let us construct the following declarations of the others:
    \begin{itemize}
        \item for every $j \in \agentsmini$ and every $a\in A'$, $p_j(a)=M_{d_i}$;
        \item for every $j, k \in \agentsmini$, $w_j(k) = w_j(i) = -2M_{d_i}$;
        \item any other value is set to $0$.
    \end{itemize}
    Here, in the optimum, for each $a\in A'$ there exists a unique agent from $\agentsmini$ assigned to it. Moreover, $i$ cannot be assigned to any activity in $A'$ as she cannot be assigned together with any other agent, hence the thesis follows.
\end{proof}

\OptSingle*
\begin{proof}
Since all weights are non-negative, assigning $i$ to their most preferred activity together with some (possibly empty) subset of $\agents \setminus \set{i}$ is strictly better in terms of the social welfare compared to assigning just them to any other activity. 
\end{proof}

\NonemptyCoalition*
\begin{proof}
    For any fixed declaration $d_i = (p_i, w_i)$, coalition $C\subseteq \agentsmini$, and $a^* \in A$, let $a' \in A \setminus\set{a^*}$ and let us set $\declared_{-i}$ as follows:
    \begin{itemize}
        \item for every $j \in C$, $p_j(a^*) = w_j(i) = M_{d_i}$;
        \item for every $k \in \agents\setminus\left(C \cup \set{i}\right)$, $p_k(a') = M_{d_i}$;
        \item any other value is set to $0$.
    \end{itemize}
    In the resulting instance, in any optimum assignment, $i$ will be assigned to $a^*$ as well as all the agents in $C$; no other agent is assigned to $a^*$ as they will be assigned to $a'$.
\end{proof}

\AloneToActivity*
\begin{proof}
    For any fixed $d_i = (p_i, w_i)$, we set $\declared_{-i}$ as follows: 
    \begin{itemize}
        \item for every $j\in\agentsmini$ and $a\in A$, $p_j(a)= -M_{d_i}$
        \item any other value is set to $0$.
    \end{itemize}
    In the resulting instance, in any optimum, $i$ cannot end up assigned to an activity with other agents, as it would lead to a negative social welfare. In conclusion, whichever activity $i$ is assigned to, including $\void$, she will be alone.
\end{proof}

\existsOPTbinaryPreferences*
\begin{proof}
    Let $\mech$ be a mechanism that outputs an optimal outcome and, in case of ties, that is, multiple optimal outcomes exist, prioritizes maximizing the utility of the agent who has stronger opinion on others, i.e.\ of the agent $j$ with higher $\max_{j' \in \agents\setminus\set{j'}}\left|w_j(j')\right|$. In case ties persist, then the mechanism can break them arbitrarily.
    
    Then, let $i$ be an arbitrary agent with truthful valuations $v_i=(p_i^t, w_i^t)$. Recall that, without loss of generality, we can assume $p_i^t(a_1) \geq p_i^t(a_2) \geq \dots \geq p_i^t(a_m)$.  Moreover, for all $k \geq m+1$, $a_k = \void$.

    Condition~\ref{NOM:sup} of \Cref{def:NOM} follows from \Cref{cor:OptBest} for any optimal mechanism and does not rely on the tie-breaking rule. We next prove Condition~\ref{NOM:inf} of \Cref{def:NOM}. 

    Let $\bar{p}=\max\set{p_i^t(a_n), 0}$ and  $C^*$ be any coalition such that
    $$ C^* \in \argmin_{{\emptyset \neq C\subseteq\agentsmini}} w^t_i(C).$$

    We first focus on the truthful declaration $v_i$ of $i$ and show that
    \begin{align}\label{eq:LBinfArbitraryWeightsBinaryPreferences}
        \inf\nolimits_{\declared_{-i}} u_i(\mech(\declared_{-i}, v_i)) \geq \min\set{p_i^t(a_m) + w_i^t(C^*), \bar{p}} \,.
    \end{align}
    
    For the sake of a contradiction, suppose that there is an optimal outcome where $i$ gets a strictly lower than $u^{min}=\min\set{p_i^t(a_m) + w_i^t(C^*), \bar{p}} $. 
    If $i$ is assigned to some activity $a\in A$ only the agents $C \subseteq \agentsmini$, where $C\neq \emptyset$, then, by definition of $C^*$ and the ordering of activities, we have $w_i^t(C) \geq w_i^t(C^*)$ and $p_i^t(a) \geq p_i^t(a_m)$. Hence, such an assignment yields utility at least $w_i^t(C^*) + p_i^t(a_m)\geq u^{min}$ for agent $i$.
    
    Therefore, to obtain a utility lower than $u^{min}$, $i$ must be assigned alone to some activity in $A \cup \set{\void}$. Assume that $i$ is assigned to $a_k$, hence, $p^t_i(a_k)< u^{min}\leq \bar{p}$.
    
    Since the outcome is optimal, it must be that $p_i^t(a_k) \geq 0$; otherwise, assigning agent $i$ to the void activity would strictly increase the social welfare. This implies that $\max\{p_i^t(a_n), 0\}= \bar{p}> p^t_i(a_k) \geq 0$, and hence $p_i^t(a_n) > 0$. Therefore, $n \leq m$, as otherwise $a_n = \void$ and $ p_i^t(a_n) =0$. Moreover, $k > n$, as $p_i^t(a_n) > p_i^t(a_k)$. 
    
    On the other hand, since there are only $n$ agents, there must exist some $h \leq n$ such that no agent is assigned to activity $a_h$. We can then reassign agent $i$ from $a_k$ to $a_h$, which would strictly increase the social welfare, as $p_i^t(a_h) \geq p_i^t(a_n) > p_i^t(a_k)$. This contradicts the optimality of the original outcome and proves \Cref{eq:LBinfArbitraryWeightsBinaryPreferences}.

    Consider now any possible declaration $d_i$ of $i$, we next show 
    \begin{align}\label{eq:UBinfArbitraryWeightsBinaryPreferences}
        \inf\nolimits_{\declared_{-i}} u_i(\mech(\declared_{-i}, d_i)) \leq \min\set{p_i^t(a_m) + w_i^t(C^*), \bar{p}}\, .
    \end{align}

    Let us construct $\declared_{-i}$ in the following way:
    \begin{itemize}
        \item for every $j, k \in C^*$, $w_j(k) = w_j(i) = M_{d_i}$ and $p_j(a_m) = 1$;
        \item for every $j \in C^*$ and $k \in \agents\setminus\left(C \cup \set{i}\right)$, $w_j(k) = w_k(j) = w_k(i) = - M_{d_i}$;
        \item any other preference and weight is set to $0$.  
    \end{itemize}
    
    By the definition of $M_{d_i}$, see Eq.~(\ref{eq:M_d_i}), if $|C^*| > 1$, in any optimal assignment, the agents in $C$ as well as $i$ must be assigned to $a_m$ while no other agent is assigned to $a_m$. If, instead, $|C^*| = 1$, i.e. $C^*$ consists of a single agent $j^*$, then assigning both $i$ and $j^*$ to $a^*$ may result in the same social welfare as assigning them to some $a' \in A \setminus \set{a_m}$ for which $p_i(a') = 1$. But $\max_{j' \in \agents\setminus \set{j^*}}\left|w_j(j')\right| = M_{d_i} > \max_{j' \in \agentsmini} \left|w_i(j')\right|$ by the definition of $M_{d_i}$, so $\mech$ will output the outcome where $i$ and $j^*$ are assigned to $a^*$.

    Thus, for any $d_i$, there exists a $\declared_{-i}$ such that $i$ is assigned to one of her least preferred activities, without loss of generality $a_m$, together with exactly all the agents in $C^*$. Hence, $\inf\nolimits_{\declared_{-i}} u_i(\mech(\declared_{-i}, d_i)) \leq p_i^t(a_m) + w_i^t(C^*)$.

    Now, it remains to show that $\inf\nolimits_{\declared_{-i}} u_i(\mech(\declared_{-i}, d_i)) \leq \bar{p}$. Let us construct another declaration of the other agents $\hat{\declared}_{-i}$:
    \begin{itemize}
        \item for every $j \in \agentsmini$, $\hat{p}_j(a_j)=1$;
        \item for every $j, k \in \agentsmini$, $\hat{w}_j(k) = \hat{w}_j(i) = -M_{d_i}$;
        \item any other value is set to $0$.
    \end{itemize}
    Here, none of the agents can be assigned to the same activity, and, since for every $j \in \agentsmini$ $\max_{j' \in \agents\setminus \set{j}}\left|\hat{w}_j(j')\right| = M_{d_i} > \max_{j' \in \agentsmini} \left|w_i(j')\right|$ by the definition of $M_{d_i}$, $\mech$ will prioritize the utility of any agent from $\agentsmini$ over the utility of $i$. Therefore, in the outcome output by $\mech$, for each $a\in \set{a_1, a_2, \dots, a_{n - 1}}$ there exists a unique agent from $\agentsmini$ assigned to it, and $i$ is assigned to some $a_h \in A \setminus \set{a_1, \dots, a_{n-1}}$ with $h \geq n$, alone, implying that agent $i$ gets a utility $p_i^t(a_k)\leq \max\set{p_i^t(a_n), 0}$. Thus, $\inf\nolimits_{\declared_{-i}} u_i(\mech(\declared_{-i}, d_i)) \leq \bar{p}$, which concludes the proof of \Cref{eq:UBinfArbitraryWeightsBinaryPreferences}.

    In conclusion, putting together \Cref{eq:LBinfArbitraryWeightsBinaryPreferences,eq:UBinfArbitraryWeightsBinaryPreferences},  Condition~\ref{NOM:inf} of \Cref{def:NOM} holds true.
    \end{proof}

\OPTisNOMnonNegativeWeights*
\begin{proof}[Proof of the case $\preferences=\nonneg$.]
Assume $\preferences=\nonneg$.

    Condition~\ref{NOM:sup} of \Cref{def:NOM} follows from \Cref{cor:OptBest}. We next show that Condition~\ref{NOM:inf} of \Cref{def:NOM} is also satisfied.

    Let $ j^* \in \argmin_{j \in \agents} w^t_i(j)$. Recall can assume $a_1$ is the most preferred while $a_m$ is the least preferred activity according to $i$'s truthful declaration.    
    We first consider the truthful report $v_i=(p^t_i, w^t_i)$ and show    \begin{align}\label{eq:LBinfNNweightsNNprefs}
       \inf\nolimits_{\declared_{-i}} u_i(\mech(\declared_{-i}, v_i)) \geq p_i^t(a_m) + w_i^t(j^*)
       \, .
    \end{align}
    Since $\weights = \nonneg$, and $a_m$ is the least preferred activity of $i$, and, by the definition of $j^*$, for any $\emptyset \neq C \subseteq \agentsmini$ and $a \in A$, it holds that $p_i^t(a) + w_i^t(C)\geq p_i^t(a_m) + w_i^t(j^*)$, as soon as $i$ is assigned to some activity with at least one other agent, she has a utility of at least $p_i^t(a_m) + w_i^t(j^*)$. 
   To obtain a lower utility, agent $i$ must be assigned alone to some, possibly void, activity. From \Cref{lemma:OptSingle}, such an activity can only be an activity in $\argmax_{a\in A \cup \set{\void}} p_i^t(a)$, which means an activity having value $p_i^t(a_1)$. If $p_i^t(a_1) \geq p_i^t(a_m) + w_i^t(j^*)$, then \Cref{eq:LBinfNNweightsNNprefs} holds true.
   
   If $p^t_i(a_1) < p_i^t(a_m) + w_i^t(j^*)$ and, in an optimum, $i$ is assigned to $a_1$ alone, then, since $\weights, \preferences = \nonneg$ and $m, n \geq 2$, there always exists an activity $a' \in A\setminus\set{a_1}$ such that in the same optimum a non-empty coalition $C \in \agentsmini$ is assigned to $a'$. But reassigning $i$ to $a'$ results in a surplus in the social welfare equal to 
   \begin{align*}
       p^t_i(a') - p^t_i(a_1) + w^t_i(C) + \sum_{j \in C}w_j(i) \\\geq p^t_i(a') + w^t_i(C) - p^t_i(a_1) \\\geq p^t_i(a_m) + w^t_i(j^*) - p_i^t(a_1) > 0,
   \end{align*}
   which means that $i$ cannot be a single agent assigned to $a_1$ in an optimum in the case when $p^t_i(a_1) < p_i^t(a_m) + w_i^t(j^*)$.
   
   In conclusion, \Cref{eq:LBinfNNweightsNNprefs} holds true.

    Let us now consider an arbitrary declaration $d_i$ of agent $i$. Let $C = \set{j^*}$ and $a^*= a_m$. By \Cref{lemma:NonemptyCoalition}, there exists a declaration of the other agents such that, in any optimal assignment, agent $i$ is assigned to activity $a_m$ together with exactly agent $j^*$.
    Hence, when declaring an arbitrary $d_i$ we have
    \begin{align}\label{eq:UBinfNNweightsNNprefs}
    \inf\nolimits_{\declared_{-i}} u_i(\mech(\declared_{-i}, d_i)) \leq p_i^t(a_m) + w_i^t(j^*) \, .
    \end{align}
    
    Putting together \Cref{eq:LBinfNNweightsNNprefs,eq:UBinfNNweightsNNprefs} we can conclude Condition~\ref{NOM:inf} of \Cref{def:NOM} holds true.
\end{proof}

\noOPTbinaryWeights*
\begin{proof}
    Consider an instance with two activities, $a_1$ and $a_2$, and two agents, $1$ and $2$. Let  $v_1=(p^t_1, w^t_1)$ be the truthful values of $1$ with $p^t_1(a_1) = 1.5$, $p^t_1(a_2) = 0$, and $w^t_1(2) = 1$.
    
    From \Cref{lemma:OptSingle} it follows that, in an optimum, the utility of $i$ cannot be less than $1$. At the same time, in the unique optimum of the instance, where the agent $2$ declares $d^*_2=(p^*_2, w^*_2)$ so that $p^*_2(a_1) = 0$, $p^*_2(a_2) = 3$, $w^*_2(1) = 1$, agent $1$ is assigned to the activity $a_2$ together with $2$. Thus, in the worst case, the utility of the agent $1$ is exactly equal to $1$.

    If $1$ declares $d_1=(p_1,w_1)$ where $p_1(a_1) = 4$ instead of $1.5$, and any other value stays the same as in $v_1$, then, in an optimum outcome, regardless the declaration $d_2=(p_2, w_2)$ of $2$, agent $1$ will always be assigned to $a_1$: since $w_2(1) \in \{0, 1\}$, $w^t_1(2) + w_2(1) + p^t_1(a_2) \leq 1 + 1 + 0 < p_1(a_1)$.
    If again agent $2$ declares $d^*_2=(p^*_2, w^*_2)$, agent $1$ will be assigned alone to activity $a_1$.
    So, in the worst case, $i$ gets utility at least equal to $1.5$, which makes the false declaration $d_i$ an obvious manipulation for $i$, as it violates Condition~\ref{NOM:inf}, implying that there does not exist an optimal NOM mechanism.
\end{proof}

\quadraticInapproximability*
\begin{proof}
Given an instance of \textsc{MaxClique} in an undirected graph $G=(V, E)$ with $n$ vertices.

We build a \gasp\ instance with only one activity $a$ and $n$ agents, one for each vertex of $G$, that is, for each $i\in V$, $i$ is an agent in the \gasp\ instance. Each agent values $1$ activity $a$. We set $w_i(j)=w_j(i)=1$, if $(i,j)\in E$, and $w_i(j)=w_j(i)=- n^3$, otherwise. Notice that in an optimal outcome, if $(i,j)\not\in E$, then $i$ and $j$ cannot be assigned to the same activity; this also holds true for any outcome with non-negative social welfare.
Therefore, in such an instance, the optimum is achieved by assigning the agents corresponding to a maximum clique of $G$ to activity $a$ and the remaining agents to the void activity. Hence, the value of the optimum is $\bar{k}\cdot (\bar{k} -1) +\bar{k} = \bar{k}^2 $, where $\bar{k}$ is the number of agents in a maximum clique of $G$. Consider now an approximation algorithm guaranteeing an approximation in $O(n^{2-\epsilon})$. Since the approximation is bounded, it means that the agents assigned to the activity must form a clique in $G$. Let $k$ be the size of such a clique, then the social welfare of the output is $k\cdot(k-1) + k =k^2$. By assumption, 
$
\frac{\bar{k}^2}{k^2} \in O(n^{2-\epsilon})\, ,
$
which means $\frac{\bar{k}}{k}\in O(n^{1-\epsilon'})$, where $\epsilon'= \epsilon/2$. This contradicts the inapproximability result for \textsc{MaxClique}.
\end{proof}

\twoapxNOM*
\begin{proof}[Proof of Condition~\ref{NOM:sup}.]
    If $i$ truthfully reports $v_i=(p_i^t, w_i^t)$, there exists a declaration of the others guaranteeing her a utility equal to $\max_{a\in A} p_i^t(a) + \sum_{j\in\agents\setminus\set{i}} w_i^t(j)$. Such a declaration $\declared_{-i}$ is defined as follows. 
    Let us fix $a^* \in \argmax_{a\in A} p_i^t(a)$, if $\argmax_{a\in A} p_i^t(a)\neq \emptyset$, and $a^*\in A$, otherwise. We set:
    \begin{itemize}
        \item $p_j(a^*)=1$, for each $j\in\agents\setminus\set{i}$;
        \item $w_j(k)= 1$, for all $j\in\agents\setminus\set{i}$, $k\in\agents$.
    \end{itemize}
Any other preference and weight is set to $0$.  
    
    For such an instance, all the agents will be assigned together with $i$ to $a^*$. Since no assignment can guarantee $i$ higher utility, no manipulation can increase her utility in the best-case scenario and Condition~\ref{NOM:sup} follows.
\end{proof}

\subsection{A modification of \Cref{mech:matching} in case $\preferences=\binary$ and $\weights=\arbitrary$}\label{appendix:binaryMechanism}

\begin{example}
When $\preferences=\binary$ and $\weights=\arbitrary$, \Cref{mech:matching}, depending on the tie-breaking rule it uses, may not be NOM.     

Consider an instance with $m = 2$, $n = 2$, and let agent $1$ have the following true preferences: $w^t_1(2) = 0$, $p^t_1(a_1) = 1$, $p^t_1(a_2) = 0$. The worst possible utility of the agent $1$ equals $0$. If the agent $2$ declares $w_2(1) = 1$, $p_2(a_2) = 1$, $p_2(a_1) = 0$, \ref{mech:matching} will have to assign both agents $1$ and $2$ either to $a_1$ or to $a_2$ since it will result into surplus in the social welfare being $2$, which is not possible when assigning a single agent to their most preferred activity. Let us assume that \ref{mech:matching} chooses the activity $a_2$ (when breaking ties), which will result in the agent $1$ having the utility $0$.

Now, let us consider the declaration $d_1$ of agent $1$ identical to the true one except for $w_1(2)$, which is declared to be $0.1$. Then, regardless of $w_2(1)$ and $p_2(1)$, being assigned to $a_1$ either together with the agent $2$ or alone is always possible in an optimum for the agent $1$: $\max\{p_1(a_1), p_1(a_1) + w_1(2) + w_2(1) + p_2(a_1)\} \geq p_1(a_2) + w_1(2) + w_2(1) + p_2(a_2)$ since $p_1(a_1) + p_2(a_1) \geq 1$ and $p_1(a_2) + p_2(a_2) \leq 1$. And if \ref{mech:matching} always chooses one of these two options, then the agent $1$ always gets the utility $1$ and thus successfully improves their worst case.
\end{example}
With the next theorem, we show that appropriately choosing the tie-breaking rule will make \Cref{mech:matching} NOM.
\begin{theorem}
If $\preferences=\binary$ and $\weights=\arbitrary$, there exists a tie-breaking rule that makes \Cref{mech:matching} NOM.
\end{theorem}
\begin{proof}
For Condition~\ref{NOM:sup}, it is not necessary to specify a particular tie-breaking rule, as the proof of \Cref{thm:matchingNOMnonNeg} remains valid in this case under binary preferences.

To prove Condition~\ref{NOM:inf}, instead, we will adapt Algorithm~\ref{algo:matching} as follows:
\begin{itemize}
    \item[] In case $x\geq y$, the $\gets\argmin$ selects an agent that is positively evaluated by at least another agent, if any exists, and breaks ties arbitrarily, otherwise.
   \item[] In case $y>x$, the $\gets\argmin$ selects a pair of agents, if more than one activity guarantees the same contribution to the welfare, the agent, in the selected pair, with the highest weight towards the other selects an activity they should be assigned to among the ones maximizing the social welfare.
\end{itemize}

{\em Condition~\ref{NOM:inf}.} 
We denote the activities in $A$ as $a_1, \dots, a_m$ and we can assume without loss of generality that $p^t_i(a_1)\geq p^t_i(a_2)\geq \dots \geq p^t_i(a_m)$. Let us set $a^* = a_m$.

Consider $d_i=(p_i,w_i)$ an arbitrary declaration of agent $i$; note that $d_i$ may possibly be $v_i$. We next show that, regardless of agent~$i$’s declaration, the worst outcome that \Cref{mech:matching} may output for this agent does not depend on her declaration. To establish this, we distinguish between the cases in which $w^t_i(j^*) + p^t_i(a^*)$ is positive or negative, since the worst outcome depends on the sign of this quantity.
    
If $w^t_i(j^*) + p^t_i(a^*) > 0$, then the worst outcome for agent~$i$ that can be produced by \Cref{mech:matching} is being assigned to the void activity, if $m \leq n - 1$, and to $a_n$ (or any activity $a$ for which $p^t_i(a)=p^t_i(a_n)$) alone, otherwise.
Given $d_i$, this can happen when the declarations of others are as follows:
\begin{itemize}
    \item for $j\in [i-1]$, $p_j(a_j) = 1$;
    \item for $j \in [m+1]\setminus [i]$, $p_{j+1}(a_j) = 1$;
    \item for $j \in \agents\setminus\set{i}$, $w_j(i) = -M_{d_i}$.
\end{itemize}
Any other preference and weight is set to $0$. Since the mechanism breaks ties against agents evaluated negatively by any other agent $i$ will never have the chance of being assigned to any activity she values $1$.

If, instead, $w^t_i(j^*) + p^t_i(a^*)  \leq 0$, the worst outcome for $i$ is to be assigned to $a^*$ together with $j^*$. Given any $d_i$, if the declarations of the other agents are such that $w_{j^*}(i) = M_{d_i}$, $p_{j^*}(a^*) = 1$, and all other preferences and weights are set to $0$. Now, since ties are broken in favor of the agent with the highest weight towards the other, then $i$ and $j^*$ will be assigned to $a^*$.

Thus, regardless of the declaration of $i$, there exists a game instance where $i$ ends up in the worst case that can possibly be output by \ref{mech:matching}. Therefore, the worst case cannot be improved, and \Cref{mech:matching} is indeed NOM.
\end{proof}

\subsection{Further Results}\label{sec:furtherResults}

\begin{restatable}{proposition}{binaryOPTnotNecNOM}
If $\preferences=\weights=\binary$, an optimal mechanism is not necessarily NOM.
\end{restatable}
\begin{proof}
Let $m=3$, $n=3$, and let agent $1$ have the following truthful  preferences and weights: $p^t_1(a_1) = 1$, $p^t_1(a_2) = p^t_1(a_3) = 0$, $w^t_1(2) = w^t_1(3) = 1$.

There exists such $\declared_{-1}$ that in an optimum $1$ gets a utility equal to $1$. For example, it happens when in $\declared_{-1}$ where $p_2(a_2) = p_3(a_3) = 1$ and all other preferences and weights are equal to zero. Here, the optimal social welfare equals $3$ and one of the optimum assignments is $z_1 = a_1$, $z_2 = a_2$ and $z_3 = a_3$. Let us assume an optimal mechanism outputs this assignment. Note that getting the utility equal to $1$ is the worst possible that can happen in an optimum to agent $1$ for any weights and preferences of others: if $1$ is the only agent assigned to $a_2$ or $a_3$, which is their only possibility to get $0$, reassigning them to $a_1$ strictly increases the social welfare.

Let us then consider another declaration $d_1=(p_1, w_1)$ of agent $1$ where $w_1(2) = 0$ and all other preferences and weights remain the same. In this case, regardless of what are preferences and weights of agents $2$ and $3$, there always exists an optimum giving agent $1$ utility at least equal to $2$ (according to their true valuations) since there always exists an optimum where agent $2$, or agent $3$, or both of them are assigned to the activity $a_1$ together with agent $1$.

Suppose for the new declaration of agent $1$ the same optimal mechanism always outputs such an outcome. Then, by misreporting their valuations agent $1$ improves their worst-case scenario and thus the mechanism is obviously manipulable.  
\end{proof}